\documentclass[10pt,twocolumn]{article}

\usepackage{geometry}
\usepackage{amsmath,amssymb,amsthm}
\usepackage{listings}
\usepackage{hyperref}
\usepackage{enumitem}
\usepackage{tikz}
\usepackage{dblfloatfix}
\usepackage{times}

\theoremstyle{plain}
\newtheorem{theorem}{Theorem}[section]
\newtheorem{lemma}[theorem]{Lemma}
\newtheorem{proposition}[theorem]{Proposition}

\theoremstyle{definition}

\theoremstyle{remark}

\title{\vspace{-1em}A Geometric Derivation of the\\
Bitner--Ehrlich--Reingold Loopless Gray Code Algorithm}
\author{Andrew Au\\[0.2em]
\small Independent Researcher\\[0.2em]
\small Corresponding author: \texttt{cshung@gmail.com}}
\date{\vspace{-1.5em}}

\begin{document}

\twocolumn[
\maketitle
\vspace{-1em}
]

\begin{abstract}
The Bitner--Ehrlich--Reingold algorithm generates the binary reflected
Gray code with constant work per codeword, using a focus-pointer array.
Its compact update is easy to state but gives little indication of why
such pointers should exist.  This note reconstructs a geometric route to
the algorithm.  The sequence of flipped bit positions is the ruler
sequence, OEIS A007814.  We realize its finite prefixes as in-order
traversals of recursively expanding trees, decorate each activation with
its nearest ancestor to the right, contract all stack operations between
successive outputs into successor jumps, and prepare the bounded
level-indexed stack by predicting its future slot values.  This
gives and proves a branch-based loopless generator.  The construction
grew from an earlier informal public exposition by the author.  Viewing horizontal
position in the tree as time then explains each stack-slot write as a
prediction for the next activation at the same level.  Finally, the two
possible future roles---an inherited continuation for a right child and
a default continuation for a left child---are scheduled together,
leading to the two assignments of the published focus-pointer algorithm.
The account is a reconstructed derivation, not a claim about the
historical reasoning of Bitner, Ehrlich, or Reingold.
\end{abstract}


\section{Introduction}

For \(n\geq1\), let the bits of an \(n\)-bit binary reflected Gray code
(BRGC) be numbered \(0,\ldots,n-1\).  The BRGC orders the vertices of the \(n\)-dimensional
hypercube along a Hamiltonian path; for \(n\geq2\), the closing edge makes
it a Hamiltonian cycle.  This note studies the induced coordinate-flip
schedule rather than the cube embedding itself.  The positions flipped
after the initial all-zero word, for \(n=4\), are
\[
0,\ 1,\ 0,\ 2,\ 0,\ 1,\ 0,\ 3,\ 0,\ 1,\ 0,\ 2,\ 0,\ 1,\ 0.
\]
Bitner, Ehrlich, and Reingold gave a loopless algorithm for producing
the BRGC \cite{ber76}.  Knuth presents its focus-pointer form as
\emph{Algorithm L} \cite[\S 7.2.1.1]{knuth4a}:

\par\noindent\begin{minipage}{\columnwidth}
\begin{lstlisting}
a[0..n - 1] = 0; f[j] = j for 0 <= j <= n
visit(a)
repeat:
    j = f[0]; f[0] = 0
    if j == n: terminate
    flip(a[j])
    f[j] = f[j + 1]; f[j + 1] = j + 1
    visit(a)
\end{lstlisting}
\end{minipage}\par

For the remainder of the paper, we analyze only the sequence of flipped
positions.  We say that Algorithm L \emph{emits} \(j\) when it executes
\lstinline!flip(a[j])!; the \lstinline!visit(a)! operations output the
corresponding Gray-code words but play no role in the derivation.

At the start of each iteration, \(f[0]\) is the next level, with \(n\)
serving as the termination sentinel; the remaining entries encode
deferred successor information.  The purpose of the derivation is to
identify that information geometrically.

Every iteration uses \(O(1)\) operations, after \(O(n)\) initialization.
The difficulty is motivational: why should the two focus-pointer writes
encode the recursive Gray-code order?  For broader background on
combinatorial generation and Gray codes, see Ruskey \cite{ruskey03}.

This note develops an answer from the geometry of the flip sequence.
The route originated in an exploratory Stack Overflow answer posted by
the author in 2021 \cite{au21}.  That answer found the decorated tree and
a branch-based loopless algorithm, but it moved too quickly from
recursive execution to successor jumps and stopped short of deriving
Algorithm L.  Here we fill those gaps.

We make no claim about how the published algorithm was historically
discovered.  Our claim is mathematical: each transformation below
preserves the generated flip sequence, and the endpoint is Algorithm L.

\section{The Ruler Sequence}
\label{sec:ruler}

Let \(S_k\) be the flip-position sequence for a \((k+1)\)-bit BRGC.
Recursive reflection gives
\[
S_{k+1}=S_k,\ k+1,\ \overleftarrow{S_k}.
\]
Here \(\overleftarrow{S_k}\) denotes \(S_k\) in reverse order.
The sequence \(S_0\) is a palindrome, and this recurrence preserves
palindromicity.  Therefore
\[
S_0=(0),\qquad S_{k+1}=S_k,\ k+1,\ S_k.
\]
Thus the finite sequences are nested prefixes of
\[
0,1,0,2,0,1,0,3,0,1,0,2,0,1,0,4,\ldots.
\]
This is the ruler sequence, OEIS A007814 \cite{oeisA007814}.

\begin{proposition}
For every \(k\geq0\),
\[
S_k=(\nu_2(1),\ldots,\nu_2(2^{k+1}-1)).
\]
Consequently, for \(t\geq1\), the \(t\)-th term of the infinite flip
sequence is
\[
\nu_2(t)=\max\{\ell:2^\ell\text{ divides }t\}.
\]
\end{proposition}

\begin{proof}
The claim is immediate for \(k=0\).  Assume it holds for \(S_k\).
For \(1\leq t<2^{k+1}\), write \(t=2^q u\), where \(u\) is odd and
\(q\leq k\).  Then
\[
2^{k+1}+t=2^q\bigl(2^{k+1-q}+u\bigr),
\]
and the parenthesized factor is odd.  Hence
\[
\nu_2(2^{k+1}+t)=q=\nu_2(t).
\]
Also \(\nu_2(2^{k+1})=k+1\).  Therefore
\[
(\nu_2(1),\ldots,\nu_2(2^{k+2}-1))
=S_k,\ k+1,\ S_k,
\]
which is exactly the recurrence for \(S_{k+1}\).
Since the nested prefixes \(S_k\) exhaust the infinite sequence, the
consequent formula holds for every \(t\geq1\).
\end{proof}

The valuation formula gives a direct arithmetic generator, but it does
not explain the focus pointers.  For that, the recursive geometry is
more revealing.

\section{An Expanding In-Order Tree}
\label{sec:tree}

Number levels upward from the leaves, beginning with level \(0\).  The
following recursive procedure generates \(S_k\).
\par\noindent\begin{minipage}{\columnwidth}
\begin{lstlisting}
def gen(level):
    if level == 0:
        print(0)
    else:
        gen(level - 1)
        print(level)
        gen(level - 1)
\end{lstlisting}
\end{minipage}\par
For the call \lstinline!gen(k)!, the activation tree is complete: a call
at level \(k>0\) has two children at level \(k-1\), and a call at level
\(0\) is a leaf.  Printing occurs in-order.  We use \emph{activation},
\emph{node}, and \emph{occurrence} interchangeably for a call represented
in this tree.

Thus \lstinline!gen(k)! emits \(S_k\).  The tree need not stop at a fixed
root: the tree rooted at level \(k\) is the left subtree of a future node
at level \(k+1\), which in turn belongs to the tree rooted at level
\(k+2\).  The direct limit generates the infinite ruler sequence, and
stopping before level \(k+1\) cuts out exactly \(S_k\).

In particular, a finite \(n\)-bit instance is rooted at level \(n-1\),
with its sentinel future parent at level \(n\).

\section{Right-Parent Decorations}
\label{sec:decoration}

For an occurrence \(v\), follow its parent chain upward in the drawing.
Its \emph{right parent} is the first ancestor reached along an edge that
moves upward and to the right; this rightward turn gives the object its
name.  Equivalently, it is the nearest ancestor whose left subtree
contains \(v\), reached on the first ascent from a left child; this is
the ancestor used in the usual in-order-successor rule.  Decorate each
node by the level of this ancestor.  The future parent above the current
root supplies the decoration inherited along the entire rightmost
path, including the rightmost leaf.  Expanding the tree does not change
an existing node's right parent: the current root already treats the
next-stage root as its parent.

Figure~\ref{fig:decorated-tree} shows the decorated tree for \(n=4\).
The horizontal coordinate is in-order time, so nodes on the same level
appear in exactly the order in which a level-indexed stack slot will
serve them.  The dashed node is the future parent that decorates the
finite tree's rightmost path.

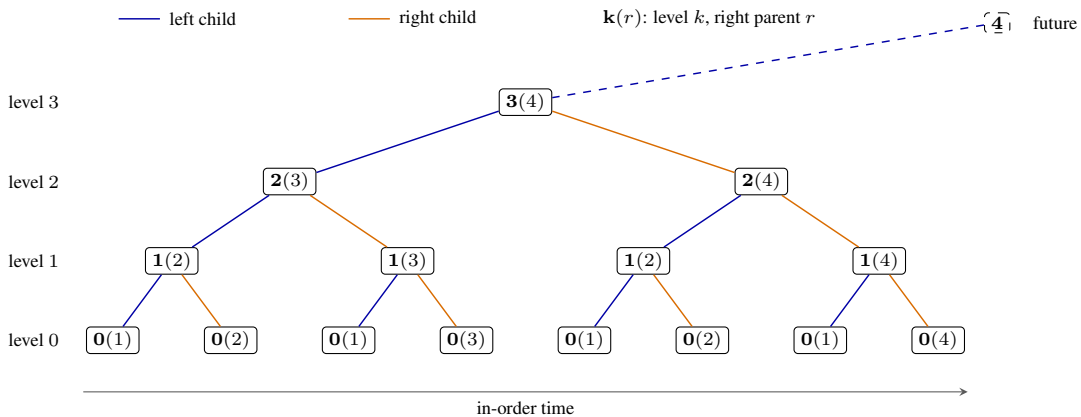
\begin{figure*}[b]
\centering
\begin{tikzpicture}[
    x=0.78cm,
    y=1.05cm,
    node/.style={
      draw,
      rounded corners=1.5pt,
      fill=white,
      inner xsep=2.5pt,
      inner ysep=1.5pt,
      font=\scriptsize
    },
    future/.style={node, dashed},
    left edge/.style={draw=blue!65!black, line width=0.55pt},
    right edge/.style={draw=orange!85!black, line width=0.55pt},
    time/.style={->, >=stealth, draw=black!65}
]

\node[node] (n01) at (1,0) {\(\mathbf{0}(1)\)};
\node[node] (n03) at (3,0) {\(\mathbf{0}(2)\)};
\node[node] (n05) at (5,0) {\(\mathbf{0}(1)\)};
\node[node] (n07) at (7,0) {\(\mathbf{0}(3)\)};
\node[node] (n09) at (9,0) {\(\mathbf{0}(1)\)};
\node[node] (n11) at (11,0) {\(\mathbf{0}(2)\)};
\node[node] (n13) at (13,0) {\(\mathbf{0}(1)\)};
\node[node] (n15) at (15,0) {\(\mathbf{0}(4)\)};

\node[node] (n12) at (2,1) {\(\mathbf{1}(2)\)};
\node[node] (n16) at (6,1) {\(\mathbf{1}(3)\)};
\node[node] (n110) at (10,1) {\(\mathbf{1}(2)\)};
\node[node] (n114) at (14,1) {\(\mathbf{1}(4)\)};

\node[node] (n24) at (4,2) {\(\mathbf{2}(3)\)};
\node[node] (n212) at (12,2) {\(\mathbf{2}(4)\)};

\node[node] (n38) at (8,3) {\(\mathbf{3}(4)\)};
\node[future] (n416) at (16,4) {\(\mathbf{4}\)};

\draw[left edge]  (n12) -- (n01);
\draw[right edge] (n12) -- (n03);
\draw[left edge]  (n16) -- (n05);
\draw[right edge] (n16) -- (n07);
\draw[left edge]  (n110) -- (n09);
\draw[right edge] (n110) -- (n11);
\draw[left edge]  (n114) -- (n13);
\draw[right edge] (n114) -- (n15);
\draw[left edge]  (n24) -- (n12);
\draw[right edge] (n24) -- (n16);
\draw[left edge]  (n212) -- (n110);
\draw[right edge] (n212) -- (n114);
\draw[left edge]  (n38) -- (n24);
\draw[right edge] (n38) -- (n212);
\draw[left edge, dashed] (n416) -- (n38);

\node[anchor=east, font=\scriptsize] at (0.25,0) {level \(0\)};
\node[anchor=east, font=\scriptsize] at (0.25,1) {level \(1\)};
\node[anchor=east, font=\scriptsize] at (0.25,2) {level \(2\)};
\node[anchor=east, font=\scriptsize] at (0.25,3) {level \(3\)};
\node[anchor=west, font=\scriptsize] at (16.45,4) {future};

\draw[time] (0.5,-0.65) -- (15.5,-0.65)
  node[midway, below, font=\scriptsize] {in-order time};

\draw[left edge] (1.1,4.05) -- (1.8,4.05)
  node[right, text=black, font=\scriptsize] {left child};
\draw[right edge] (5.0,4.05) -- (5.7,4.05)
  node[right, text=black, font=\scriptsize] {right child};
\node[anchor=west, font=\scriptsize] at (9.15,4.05)
  {\(\mathbf{k}(r)\): level \(k\), right parent \(r\)};
\end{tikzpicture}
\caption{The level-\(3\) tree as the left subtree of a future level-\(4\)
node.  Horizontal position is traversal time.}
\label{fig:decorated-tree}
\end{figure*}

The decorations arise directly from the recursive calls:
\par\noindent\begin{minipage}{\columnwidth}
\begin{lstlisting}
def gen(level, right_parent):
    if level == 0:
        emit(0, right_parent)
    else:
        gen(level - 1, level)          # left
        emit(level, right_parent)
        gen(level - 1, right_parent)   # right
\end{lstlisting}
\end{minipage}\par
For a tree rooted at level \(k\), call this procedure as
\lstinline!gen(k, k + 1)!; the second argument is supplied by the future
level-\((k+1)\) parent.
A left child receives its parent level; a right child inherits its
parent's own right parent.

\section{Contracting the Recursive Control Flow}
\label{sec:contraction}

The recursive generator already takes \(O(1)\) amortized work per output.
Every activation prints exactly once.  Charge the entry and exit of each
activation to that print; every output is then charged exactly one push
and one pop.  This aggregate bound does not give constant delay, however,
because several returns and calls may still occur between two particular
outputs.

We do not need to simulate those silent operations.  They form a path
from one printed node to its in-order successor and may be contracted
into a single successor edge.

\begin{lemma}[Successor contraction]
\label{lem:successor}
Starting at the leftmost leaf, repeatedly moving directly to the
in-order successor visits nodes in exactly the order emitted by the
recursive generator.
\end{lemma}

\begin{proof}
This is the defining property of in-order traversal.  If a visited node
has a right subtree, its successor is the leftmost node of that subtree.
Otherwise its successor is its right parent.  Contracting the intervening
tree edges changes no visited node and therefore no output.
\end{proof}

For the finite \(n\)-bit tree, the rule simplifies further:
\begin{itemize}[nosep]
  \item after a nonzero node, the successor is the leftmost leaf of its
        right subtree, hence has label \(0\);
  \item after a leaf \(0(r)\), the successor has label \(r\), unless
        \(r=n\), in which case this leaf is the final output of the
        finite \(n\)-bit tree.
\end{itemize}
The contraction identifies the next output from the current node and its
decoration.  It remains to prepare the recursive stack without simulating
the intervening pushes and pops.

\section{Predicting the Stack}
\label{sec:prediction}

From this point onward, fix an \(n\)-bit Gray code.  Its finite traversal
is \lstinline!gen(n - 1, n)!, with real levels \(0,\ldots,n-1\) and
future sentinel level \(n\).  A root-to-leaf recursion stack therefore
contains at most \(n\) activation records, with at most one activation at
each level.  Each recursive descent decreases \lstinline!level! by one,
so after indexing the stack by level, an activation with value \(j\)
always occupies logical stack slot \(j\).  The level and logical stack
position are the same quantity.

Represent this fixed-height stack by an array \(\mathrm{rp}\).
Entry \(\mathrm{rp}[0]\) stores the current logical stack position, and
\(\mathrm{rp}[j+1]\) stores the right-parent parameter for level \(j\),
where \(0\leq j<n\).
Thus the right-parent value for level \(j\) is stored in
\(\mathrm{rp}[j+1]\).  This one-position offset aligns the array with
Algorithm L from the outset.  Entry \(\mathrm{rp}[0]\) is exceptional:
it stores the next emitted level, not a decoration.

This array remains the logical recursion stack.  What changes is how its
values are prepared: instead of producing them through LIFO pushes and
pops, we predict them directly.  After a level-\(j\) occurrence is
emitted, we may write into its slot the right-parent value that the
\emph{next} level-\(j\) occurrence will need, even if that occurrence is
many outputs in the future.
Accordingly, the invariant is that \(\mathrm{rp}[j+1]\) contains the
decoration of the next not-yet-emitted level-\(j\) occurrence; when
level \(j\) is current, this is the current occurrence's decoration.
Initialize \(\mathrm{rp}[i]=i\), as on the initial descent through left
children; in particular, \(\mathrm{rp}[0]=0\) is the initial current
level.  This gives the branch-based generator immediately:
\par\noindent\begin{minipage}{\columnwidth}
\begin{lstlisting}
rp[i] = i for 0 <= i <= n
while rp[0] != n:
    j = rp[0]
    emit(j)
    r = rp[j + 1]
    if r == j + 1:
        if r < n:
            rp[j + 1] = rp[j + 2]
    else:
        rp[j + 1] = j + 1
    rp[0] = r if j == 0 else 0
\end{lstlisting}
\end{minipage}\par

Here the comparison chooses between two predictions for the next value
of stack slot \(j\):
\begin{align}
r=j+1<n &\implies \mathrm{rp}[j+1]\gets\mathrm{rp}[j+2],
                                                      \tag{R1}\\
r>j+1 &\implies \mathrm{rp}[j+1]\gets j+1.            \tag{R2}
\end{align}
At the root, \(r=j+1=n\); there is no later level-\((n-1)\)
occurrence to predict, so no stack-slot update is required.

\begin{lemma}
\label{lem:left-right}
A level-\(j\) occurrence has decoration \(j+1\) exactly when it is a
left child of level \(j+1\); a right child has decoration greater than
\(j+1\).
\end{lemma}

\begin{proof}
A left child receives \(j+1\) by definition.  A right child inherits
the decoration of its level-\((j+1)\) parent.  A decoration is the level
of a strict ancestor, so the parent's decoration is at least \(j+2\).
This includes the current root, which is a left child of its
future parent.  Hence the two cases \(j+1\) and \(>j+1\) are exhaustive.
\end{proof}

In Figure~\ref{fig:decorated-tree}, vertical position is level and
horizontal position is in-order time.  Each assignment to
\(\mathrm{rp}[j+1]\) sends information horizontally to the next occurrence
on level \(j\):
\[
\begin{aligned}
&\text{current level-\(j\) occurrence}\\[-0.25em]
&\qquad\xrightarrow{\quad\mathrm{rp}[j+1]\quad}
  \text{next level-\(j\) occurrence}.
\end{aligned}
\]
That next occurrence may be many outputs away.

\begin{theorem}
\label{thm:rp}
The branch-based generator emits \(S_{n-1}\).  It uses \(O(1)\) work per
output.
\end{theorem}

\begin{proof}
Initially, the first occurrence at every level is a left child, including
the root as a child of future level \(n\), so
\(\mathrm{rp}[j+1]=j+1\) establishes the future-decoration invariant.

When level \(j\) is emitted, first read \(r=\mathrm{rp}[j+1]\); it is the
decoration of the current occurrence.  Suppose that occurrence is a left
child below the root.  The next occurrence at the same level is
its right sibling, which inherits the parent's decoration.  That parent
has not yet been emitted: it is the next occurrence at level \(j+1\),
whose decoration is \(\mathrm{rp}[j+2]\).  Assigning this value to
\(\mathrm{rp}[j+1]\) re-establishes the invariant for the next
level-\(j\) occurrence; this is R1.  If the occurrence is the finite
root, no later occurrence exists at that level, so skipping the update
preserves every value that will still be read.

If the current occurrence is a right child, the next level-\(j\)
occurrence lies in a later sibling group and is a left child, whose
decoration is \(j+1\); this is R2.  Lemma~\ref{lem:left-right} shows that
the comparison selects exactly these cases.  Only \(\mathrm{rp}[j+1]\) is
written among the decoration slots, so every other level's invariant,
including the \(\mathrm{rp}[j+2]\) value read by R1 when \(j<n-1\),
remains intact.

Finally, Lemma~\ref{lem:successor} gives the next emitted level: after a
leaf it is its decoration \(r\), and after a nonzero node it is the
leftmost leaf \(0\).  The invariant and the emitted sequence therefore
agree with the recursive traversal.  The unique leaf decorated \(n\) is
the rightmost leaf of the level-\((n-1)\) finite tree.  Its successor is
the future parent \(n\), so the termination test stops after exactly
\(2^n-1\) outputs, the prefix \(S_{n-1}\).
\end{proof}

Without the test \(\mathrm{rp}[0]=n\), the same construction can continue
indefinitely if stack slots are created lazily with default
\(\mathrm{rp}[i]=i\).  After \(t\) outputs only \(O(\log t)\) levels
have been reached.

\section{From Stack Slots to Focus Pointers}
\label{sec:focus}

We now transform the branch-based generator into Algorithm L.  Both
programs use \(j\) for the current level in slot \(0\).  The branch-based
machine additionally reads \(r=\mathrm{rp}[j+1]\), the current
occurrence's right parent.

The array indexing already agrees.  Rename \(\mathrm{rp}\) to \(f\):
\[
f[i]\leftrightarrow\mathrm{rp}[i]
\quad(0\leq i\leq n).
\]
The arrays now have exactly the same bounds and storage layout.  What
changes is only when the stack-slot updates are performed.

For the branch-based machine, a left-child occurrence below the finite
root performs R1, while a right-child occurrence performs R2.  The
root needs no same-level prediction.  Occurrences at each other
level alternate left, right, left, right.  The branch can be removed by
changing when these two assignments occur.

\paragraph{Delay R1 to the parent.}
Between a left level-\(j\) occurrence below the root and its right
sibling, their level-\((j+1)\) parent is visited.  The left occurrence's
R1 write is not needed until the right sibling appears, so it may be
delayed to that parent visit.  Its source is stable as well: only nodes
below level \(j\) are visited before that parent, and such visits cannot
write the source slot for level \(j+1\).  At a current parent level
\(j\), the delayed R1 for level \(j-1\) is
\[
f[j]\gets f[j+1].
\]
Every positive-level node has a left child, so this delayed assignment
always serves an actual right sibling.

\paragraph{Perform R2 speculatively.}
At every level-\(j\) visit, perform
\[
f[j+1]\gets j+1
\]
unconditionally.  On a right child, this is exactly the required R2
reset.  On a left child, it is premature but harmless: the intervening
parent visit performs the delayed R1 write and overwrites it before the
right sibling can read the slot.  After a right child, no level-\((j+1)\)
visit occurs before the next left level-\(j\) occurrence, so the reset
survives until it is needed.  The root is the sole exception to
the overwrite: \(f[n]\gets n\) is a no-op, and its right sibling lies
beyond the finite traversal, so the value is never read again.

Consequently every positive level-\(j\) visit performs the fixed pair
\[
\underbrace{f[j]\gets f[j+1]}_{\text{delayed R1 for level }j-1},
\qquad
\underbrace{f[j+1]\gets j+1}_{\text{speculative R2 for level }j}.
\]
The assignments occur in this order because the first must copy the old
value of \(f[j+1]\) before the second resets it.  Except at the finite
root, the branch-based machine executes one of R1 and R2 for its current
level; Algorithm L instead executes a delayed R1 for the level below and
an unconditional R2 for the current level.  Thus this is a rescheduling,
plus harmless speculative writes, rather than merely exchanging two
source lines.

\paragraph{Absorb the successor update.}
At the start of an iteration, read \(j=f[0]\) and reset \(f[0]\) to zero.
If \(j>0\), the two focus writes do not touch \(f[0]\), so the next level
is \(0\), exactly as in the branch-based machine.  If \(j=0\), there is
no lower level awaiting a delayed R1, and \(f[0]\) is instead
the current-level slot.  The same instruction therefore becomes
\[
f[0]\gets f[1],
\]
which selects the leaf's right parent as the next level; the second write
resets \(f[1]\) to \(1\).  The old \(f[1]\) is read before that reset.
These are exactly the Algorithm L updates displayed in the introduction.

\begin{lemma}[Correctness of the rescheduling]
\label{lem:rescheduling}
With \(f[i]=i\) initially, Algorithm L visits the same decorated tree
nodes, in the same order, as the branch-based generator.
\end{lemma}

\begin{proof}
We argue in chronological in-order time, assuming the invariant at every
earlier visit: whenever level \(j\) is current, \(f[j+1]\) is its
right-parent decoration.  For the first occurrence at level \(j\), no
earlier visit can have changed \(f[j+1]\): its only writers are
level-\(j\) R2 and level-\((j+1)\) delayed R1, and neither level has yet
appeared.  Thus the slot still has its initial value \(j+1\), the
decoration of that first left child.

Thereafter, consider two consecutive occurrences at level \(j\).  If
the first is a left child, its speculative reset is overwritten at the
intervening level-\((j+1)\) parent by
\(f[j+1]\gets f[j+2]\).  At that parent visit, \(f[j+2]\) is the
parent's decoration by the chronological induction hypothesis, so the
right sibling receives the inherited decoration required by R1.  If the
first occurrence is a right child, its reset
\(f[j+1]\gets j+1\) survives until the next level-\(j\) occurrence,
which is a left child; this is R2.  These are the only assignments that
write \(f[j+1]\), so the invariant is preserved.

Finally, the update of \(f[0]\) follows the successor rule: a nonzero
node is followed by level \(0\), while a leaf is followed by its
right-parent decoration.  Hence the branch-free schedule follows the
same in-order successor at every step.  The final leaf places \(n\) in
\(f[0]\), causing Algorithm L to terminate at the same boundary.
\end{proof}

\begin{theorem}
Algorithm L emits \(S_{n-1}\), and hence flips exactly the positions required
by the \(n\)-bit BRGC.
\end{theorem}

\begin{proof}
By Lemma~\ref{lem:rescheduling}, Algorithm L emits the same sequence as
the branch-based generator.  Theorem~\ref{thm:rp} identifies that
sequence as \(S_{n-1}\), which Section~\ref{sec:ruler} identifies as the
\(n\)-bit BRGC flip positions.
\end{proof}

The transformation is therefore direct: retain the stack-array layout,
delay R1 to the parent, perform R2 speculatively, and merge the successor
state into slot \(f[0]\).

\section{Scope and Conclusion}

The derivation proceeds through distinct, checkable transformations:
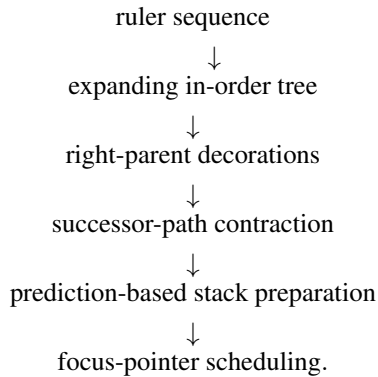
\[
\begin{gathered}
\text{ruler sequence}\\
\mathrel{\phantom{\longrightarrow}}\downarrow\\[-0.4em]
\text{expanding in-order tree}\\
\downarrow\\[-0.4em]
\text{right-parent decorations}\\
\downarrow\\[-0.4em]
\text{successor-path contraction}\\
\downarrow\\[-0.4em]
\text{prediction-based stack preparation}\\
\downarrow\\[-0.4em]
\text{focus-pointer scheduling}.
\end{gathered}
\]
The central geometric insight is temporal.  A write to stack slot \(j\)
predicts a parameter for the next occurrence on that horizontal level,
possibly many outputs in the future.  The branch-based
algorithm makes that prediction explicitly; Algorithm L stores both
possible future roles in adjacent focus slots and updates them together.

This is not offered as the historical discovery path of
Bitner--Ehrlich--Reingold.  It is a reconstruction motivated by the
author's earlier public exploration \cite{au21}, completed here with the
successor-contraction argument, the future-stack-slot invariant, and a
separate correctness proof for focus-pointer scheduling.

\noindent\textbf{Funding.}
This research received no specific grant from funding agencies in the
public, commercial, or not-for-profit sectors.

\end{document}